\documentclass{llncs}
\usepackage{amsmath,amssymb,latexsym}
\usepackage{code}
\usepackage{mathpartir}
\usepackage{upgreek}
\usepackage[T1]{fontenc}
\usepackage{alltt}
\usepackage{tikz}

\usepackage{float}
\floatstyle{boxed} 
\restylefloat{figure}

\let\inst=\undefined

\newcommand{\ie}{\emph{i.e.}}
\newcommand{\eg}{\emph{e.g.}}

\newcommand{\etal}{\emph{et al.}}

\newcommand{\var}[1]{\mathit{#1}}
\newcommand{\s}[1]{\mathit{#1}}

\newcommand{\dom}{\var{dom}}

\newcommand{\wt}{\sqsubseteq}

\newenvironment{grammar}{\begin{array}{r@{\;}c@{\;}l@{\;}c@{\;}l@{\;\;\;\;}l}}{\end{array}}

\newcommand{\opor}{\mathrel{|}}

\newcommand{\produces}{\mathrel{::=}}

\newcommand{\ttlp}{\mbox{\tt (}}
\newcommand{\ttrp}{\mbox{\tt )}}
\newcommand{\ttlc}{\mbox{\tt \{}}
\newcommand{\ttrc}{\mbox{\tt \}}}

\newcommand{\inst}{\iota}

\newcommand{\lab}{\ell}

\newcommand{\expr}{e}

\newcommand{\state}{\varsigma}

\newcommand{\store}{\sigma}

\newcommand{\addr}{a}

\newcommand{\sa}[1]{\widehat{\mathit{#1}}}

\newcommand{\astate}{{\hat{\varsigma}}}

\newcommand{\astore}{{\hat{\sigma}}}

\newcommand{\aaddr}{{\hat{\addr}}}
\newcommand{\aalloc}{{\widehat{alloc}}}

\newcommand{\absmap}{\alpha}

\input{local-macros}

\newcommand{\mylongtitle}{Pushdown Abstractions of JavaScript}

\newcommand{\myshorttitle}{\mylongtitle}

\begin{document}
\def\inst#1{\unskip$^{#1}$}
\title{\mylongtitle}
\titlerunning{\myshorttitle}
\author{David Van Horn\inst{1} \and Matthew Might\inst{2}}
%
%
\institute{Northeastern University, Boston, Massachusetts, USA
\and University of Utah, Salt Lake City, Utah, USA}

\maketitle              

\begin{abstract}
  We design a family of program analyses for JavaScript that
make \emph{no approximation} in matching calls with returns,
exceptions with handlers, and breaks with labels.  We do so by
starting from an established reduction semantics for JavaScript and
systematically deriving its intensional abstract interpretation.  Our
first step is to transform the semantics into an equivalent low-level
abstract machine: the JavaScript Abstract Machine (JAM).
We then give an infinite-state yet decidable pushdown machine whose
stack precisely models the structure of the concrete program stack.
The precise model of stack structure in turn confers \emph{precise}
control-flow analysis even in the presence of control effects, such as
exceptions and finally blocks.
We give pushdown generalizations of traditional forms of analysis such
as $k$-CFA, and prove the pushdown framework for abstract
interpretation is sound and computable.

\end{abstract}

\newcommand\evf{{\mathit{eval}}}
\newcommand\JS{{\mathit{JS}}}
\newcommand\AJS{{\widehat{\mathit{JS}}}}
\newcommand\betavalue{{\mathbf{v}}}




\section{Introduction}

JavaScript is the dominant language of the web, making it the most ubiquitous
programming language in use today.  Beyond the browser, it is increasingly
important as a general-purpose language, as a server-side scripting language,
and as an embedded scripting language---notably, Java 6 includes
support for scripting applications via the {\tt javax.script} package, and the
JDK ships with the Mozilla Rhino JavaScript engine.  Due to its ubiquity,
JavaScript has become the target language for an array of compilers for
languages such as C\#, Java, Ruby, and others, making JavaScript a widely used
``assembly language.''
As JavaScript cements its foundational role, the importance of robust static
reasoning tools for that foundation grows.

Motivated by the desire to handle non-local control effects such as
exceptions and {\tt finally} precisely, we will depart from standard
practice in higher-order program analysis to derive an
\emph{infinite}-state yet decidable pushdown abstraction from our
original abstract machine.
The stack of the pushdown abstract interpreter \emph{exactly} models the
stack of the original abstract machine with no loss of
structure---approximation is inflicted on only the control states.
This pushdown framework offers a degree of precision in reasoning about control
inaccessible to previous analyzers.

Pushdown analysis is an alternative paradigm for the analysis of
higher-order programs in which the run-time program stack is precisely
modeled with the stack of a pushdown
system~\cite{dvanhorn:Vardoulakis2011CFA2,dvanhorn:Earl2010Pushdown}.
Consequently, a pushdown analysis can exactly match control flow
transfers from calls to returns, from throws to handlers, and from
breaks to labels.  This in contrast with the traditional approaches of
finite-state abstractions which necessarily model the control stack
with finite bounds.

As an example demonstrating the basic difference between traditional
approaches, such as 0CFA, and our pushdown approach, consider the
following JavaScript program:
\begin{alltt}
   // \((\mathbb{R}\rightarrow\mathbb{R}) \rightarrow (\mathbb{R}\rightarrow\mathbb{R})\)
   // Compute an approximate derivative of f.
   function deriv(f) \{
     var \(\epsilon\) = 0.0001;
     return function (x) \{
       return (f(x+\(\epsilon\)) - f(x-\(\epsilon\))) / (2*\(\epsilon\));
     \};
   \};
   deriv(function (y) \{ return y*y; \});
\end{alltt}
The {\tt deriv} program computes an approximation to the derivative of
its argument.  In this example, it is being applied the square
function, so it returns an approximation to the double function.

It is important to take note of the two distinct calls to {\tt f}.
Basic program analyses, such as 0CFA, will determine that the square
function is the target of the call at {\tt f(x+\(\epsilon\))}.
\emph{However, they cannot determine whether the call to {\tt
    f(x+\(\epsilon\))} should return to {\tt f(x+\(\epsilon\))} or to
  {\tt f(x-\(\epsilon\))}.}  Context-sensitive analysis, such as 1CFA,
can reason more precisely by distinguishing the analysis of each call
to {\tt f}, however such methods come with a prohibitive computational
cost~\cite{dvanhorn:VanHorn-Mairson:ICFP08} and, more fundamentally,
$k$-CFA will only suffice for precisely reasoning about the control
stack up to a fixed calling context depth.

This is the fundamental shortcoming of traditional approaches to
higher-order program analysis, both in functional and object-oriented
languages.  This is an unfortunate situation, since the dominant
control mechanism is calls and returns.  To make matters worse, in
addition to higher-order functions, JavaScript includes sophisticated
control mechanisms further complicating and confounding analytic
approaches.

To overcome this shortcoming we use a \emph{pushdown} approach to
abstraction that exactly captures the behavior of the control stack.
We derive the pushdown analysis as an abstract interpretation of an
abstract machine for JavaScript.  The crucial difference between our
approach and previous approaches is that we will leave the stack
unabstracted. As this abstract interpretation ranges over an infinite
state-space, the main technical difficulty will be recovering
decidability of reachable states.

\subsection*{Challenges from JavaScript}
 JavaScript is an expressive, aggressively dynamic, high-level
programming language.  
It is a higher-order, imperative, untyped language that is both functional and
object-oriented, with prototype-based inheritance, constructors, non-local
control, and a number of semantic quirks.
Most quirks simply demand attention to detail, e.g.:
\begin{code} 
  if (false) \{ var x ; \} 
  ... x ... // x is defined \end{code}
Other quirks, such as the much-maligned {\tt with}
construct end up succumbing to an unremarkable desugaring.
Yet other features, like non-local control effects and prototypical inheritance,
require attention in the mechanics of the analysis itself; 
for a hint of what is possible, consider:
\begin{code}
 out: while (true) 
  try \{
    break out ;
  \} finally \{
    try \{
      return 10 ;
    \} finally \{ 
      console.log("this runs; 10 returns") ;
    \}
  \}\end{code}

It has become customary when reasoning about JavaScript to assume
well-behavedness---that some subset of its features are never
(or rarely) used for many programs.
Richards, Lebresne, Burg and Vitek's
thorough study~\cite{dvanhorn:Richards2010Analysis} has cast empirical doubt on these well-behavedness assumptions,
finding almost every language feature used in almost every program in a large
corpus of widely deployed JavaScript code.

Our goal is a principled approach for reasoning about \emph{all} of
JavaScript, including its unusual semantic peculiarities and its complex control
mechanisms.
To make this possible, the first step is the calculation of an \emph{abstractable}
abstract machine from an established semantics for JavaScript.
From there, a pushdown abstract interpretation of that machine yields
a sound, robust framework for static analysis of JavaScript with
precise reasoning about the control stack.

\subsection*{Contributions}

The primary contribution of this work is a provably sound and
computable framework for infinite-state pushdown abstract
interpretations of all of JavaScript, sans {\tt eval}, that
\emph{exactly} models the program stack including complex local and
non-local control-flow relationships, such as proper matching between
calls and returns, throws and handlers, and breaks and
labels.\footnote{ One might wonder why {\tt break} to a label requires
  non-local reasoning.
  In fact, it should not require it, but Guha~\etal's desugaring into $\lambda_{JS}$,
  handles {\tt break} using powerful escape continuations.
  Constraints in the desugaring process prevent these escape continuations from crossing interprocedural 
  boundaries, but unrestricted---or optimized $\lambda_{JS}$---may violate these constraints.
  For completeness, we handle full, unrestricted $\lambda_{JS}$, which means we must model
  these general escape continuations.
}

In support of our primary contribution, our secondary contributions
include the development of a variant of a known formal model for
JavaScript as a calculus of explicit substitutions; a correct abstract
machine for this model obtained via a detailed derivation, carried out
in SML, going from the calculus to the machine via small,
meaning-preserving program transformations; and executable semantic
models for the reduction semantics, the abstract machine and its
pushdown-abstractions, written in PLT
Redex~\cite{dvanhorn:Felleisen2009Semantics}.



\subsection*{Outline}

Section~\ref{sec:rho} gives the syntax and semantics of a core
calculus of explicit substitutions based on the
$\lambda_\JS$-calculus.  This new calculus, $\lambda\rho_\JS$, is
shown to correspond with $\lambda_\JS$.  Section~\ref{sec:jam} derives
an abstract machine, the JavaScript Abstract Machine (JAM), from the
calculus of explicit substitutions, which is a correct machine for
evaluating $\lambda_\JS$ programs.  The machine has been crafted in
such a way that it is suitable for approximation by a pushdown
automaton.  Section~\ref{sec:pda} yields a family of pushdown abstract
interpreters by a simple store-abstraction of the JAM.  The framework
is proved to be sound and computable.  Specific program analyses are
obtained by instantiating the allocation strategy for the machine and
examples of strategies corresponding to pushdown generalization of
known analyses are given.  Section~\ref{sec:related} relates this work
to the research literature and section~\ref{sec:conclusion} concludes.

\paragraph{Background and notation:} We assume a basic familiarity with
reduction semantics and abstract machines.  For background on
concepts, terminology, and notation employed in this paper, we refer
the reader to \emph{Semantics Engineering with PLT
  Redex}~\cite{dvanhorn:Felleisen2009Semantics}.  Our construction of
machines from reduction semantics follows Danvy, \etal's
refocusing-based
approach~\cite{dvanhorn:Danvy-Nielsen:RS-04-26,dvanhorn:Biernacka2007Concrete,dvanhorn:Danvy:AFP08}.
Finally, for background on systematic abstract interpretation of 
abstract machines, see our recent work on the approach~\cite{dvanhorn:VanHorn2010Abstracting}.

\section{A calculus of explicit substitutions: $\boldsymbol{\lambda\rho_\JS}$}
\label{sec:rho}

Our semantics-based approach to analysis is founded on abstract
machines, which give an idealized characterization of a low-level
language implementation.
As such, we need a correct abstract machine for JavaScript.
Rather than design one from scratch and manually verify its correctness after
the fact, we rely on the syntactic correspondence between calculi and
machines and adopt the refocusing-based approach of Danvy,~\etal, to
construct an abstract machine systematically from an established semantics for
JavaScript.

Guha, Saftoiu, and Krishnamurthi~\cite{dvanhorn:Guha2010Javascript} give a 
small core calculus, $\lambda_\JS$, with a small-step reduction semantics
using evaluation contexts and demonstrate that full JavaScript can be desugared
into $\lambda_\JS$.  The semantics accounts for all of JavaScript's features
with the exception of {\tt eval}.
Only some of JavaScript quirks are modeled directly, while other aspects are
treated traditionally.  For example, lexical scope is modeled with
substitution.
The desugarer is modeled formally and also available as a standalone Haskell program.

We choose to adopt the $\lambda_\JS$ model since its small size
results in a tractably sized abstract machine.

The remainder of this paper focuses on machines and abstract
interpretation for $\lambda_\JS$.  We refer the reader to Guha,~\etal,
for details on desugaring JavaScript to $\lambda_\JS$ and rational
for the design decisions made.

\subsection{Syntax}

The syntax of $\lambda\rho_\JS$ is given in figure~\ref{fig:syntax-rho}.
Syntactic constants include strings, numbers, addresses, booleans, the
undefined value, and the null value.  Addresses are first-class values
used to model mutable references.  Heap allocation and dereference is
made explicit through desugaring to $\lambda_\JS$.  Syntactic values
include constants, function terms, and records.  Records are keyed by
strings and operations on records are modeled by functional update,
extension, and deletion.  Expressions include variables, syntactic
values, and syntax for let binding, function application, record
dereference, record update, record deletion, assignment, allocation,
dereference, conditionals, sequencing, while loops, labels, breaks,
exception handlers, finalizers, exception raising, and application of
primitive operations.
A program is a closed expression.

\begin{figure}
\[
\begin{grammar}
\mstr &\in & \s{String}\\
\mnum & \in & \s{Number}\\
\maddr & \in & \s{Address}\\
\mvar & \in & \s{Variable}
\\[1mm]
\mexp,\mexpo,\mexpoo &\produces&
\mvar 
\opor \mstr 
\opor \mnum 
\opor \maddr
\opor \strue 
\opor \sfalse
\opor \sundefn
\opor \snull
\\
&\opor&
\sfunc\mvars\mexp 
\opor \srec{\vec{\scol\mstr\mexp}}  
\opor \slet\mvar\mexp\mexp 
\opor \sapp\mexp\mexps 
\opor \saref\mexp\mexp
\\
&\opor& \saset\mexp\mexp\mexp
\opor \sdel\mexp\mexp
\opor \sset\mexp\mexp 
\opor \sref\mexp 
\opor \sderef\mexp
\\
&\opor & \sif\mexp\mexp\mexp \opor \sseq\mexp\mexp \opor \swhile\mexp\mexp 
\\
&\opor & \slab\mlab\mexp \opor \sbreak\mlab\mexp
\opor \stryc\mexp\mvar\mexp
\\
&\opor & \stryf\mexp\mexp 
\opor \sthrow\mexp 
\opor \sop{\vec\mexp}
\\[1mm]
\mvaloo,\mvalo,\mval &\produces & \mstr 
\opor \mnum 
\opor \maddr
\opor \strue 
\opor \sfalse
\opor \sundefn \opor \snull
\opor \scls{\sfunc\mvars\mexp}\menv 
\opor \srec{\vec{\scol\mstr\mval}}
\\[1mm]
\mcls,\mclso &\produces & \scls\mexp\menv
\opor \srec{\vec{\scol\mstr\mcls}}
\opor \slet\mvar\mcls\mcls
\opor \sapp\mcls\mclss
\opor \saref\mcls\mcls
\\
&\opor& \saset\mcls\mcls\mcls
\opor \sdel\mcls\mcls
\opor \sset\mcls\mcls 
\opor \sref\mcls 
\opor \sderef\mcls
\\
&\opor & \sif\mcls\mcls\mcls 
\opor \sseq\mcls\mcls 
\opor \swhile\mcls\mcls
\\
&\opor & \slab\mlab\mcls 
\opor \sbreak\mlab\mcls
\opor \stryc\mcls\mvar\mcls
\\
&\opor & \stryf\mcls\mcls
\opor \sthrow\mcls 
\opor \sop{\mclss}
\end{grammar}
\]
\caption{Syntax of $\lambda_\JS$}
\label{fig:syntax-rho}
\end{figure}

\subsection{Semantics}

Guha,~\etal, give a substitution-based reduction semantics formulated
in terms of Felleisen-Hieb-style evaluation
contexts~\cite{dvanhorn:Felleisen2009Semantics}.
The use of substitution in $\lambda_\JS$ is traditional from a
theoretical point of view, and is motivated in part by want of
conventional reasoning techniques such as subject reduction. 
On the other hand, environments are traditional from an implementation
point of view.
To mediate the gap, we first develop a variant of $\lambda_\JS$ that
models the meta-theoretic notion of substitution with explicit
substitutions.

Substitutions are represented explicitly with \emph{environments},
which are finite maps from variables to values.

Substitution $[\mval/\mvar]\mexp$, which is a meta-theoretic notation
denoting $\mexp$ with all free-occurrences of $\mvar$ replaced by
$\mval$, is represented at the syntactic level in $\lambda\rho_\JS$, a
calculus of explicit substitutions, as a pair consisting of $\mexp$
and an environment representing the substitution:
$\scls\mexp{\{(\mvar,\mval)\}}$.  Such a pair is known as a \emph{closure}.
 
The heap is modeled as a top-level \emph{store}, a finite map from addresses
to values.

\begin{figure*}
\[
\begin{array}{rcl}
\multicolumn{3}{l}{\mbox{Context-insensitive, store-insensitive rules:}}
\\[1mm]
\scls\mvar\menv &\stepsto& \menv(\mvar)
\\
\slet\mvar\mval{\scls\mexp\menv} &\stepsto& \scls\mexp{\menv[\mvar\mapsto\mval]}
\\
\sapp{\scls{\sfunc\mvars\mexp}\menv}\mvals &\stepsto& \scls\mexp{\menv[\mvars\mapsto\mvals]}
\mbox{, if }|\mvars| = |\mvals|
\\
\saref{\srec{\vec{\scol\mstr\mval},\scol{\mstr_i}\mval,\vec{\scol\mstr\mval}'}}{\mstr_i} &\stepsto& \mval
\\
\saref{\srec{\vec{\scol\mstr\mval}}}{\mstr_x} &\stepsto& \sundefn
\mbox{, if }\mstr_x\not\in\vec\mstr
\\
\saset{\srec{\vec{\scol\mstr\mval},\scol{\mstr_i}{\mval_i},\vec{\scol\mstr\mval}'}}{\mstr_i}\mval &\stepsto &
\srec{\vec{\scol\mstr\mval},\scol{\mstr_i}\mval,\vec{\scol\mstr\mval}'}
\\
\saset{\srec{\vec{\scol\mstr\mval}}}{\mstr_x}\mval &\stepsto &
\srec{\vec{\scol\mstr\mval}}\text,
\mbox{ if }\mstr_x\not\in\vec\mstr
\\
\sdel{\srec{\vec{\scol\mstr\mval},\scol{\mstr_i}{\mval_i},\vec{\scol\mstr\mval}'}}{\mstr_i}
& \stepsto &
\srec{\vec{\scol\mstr\mval},\vec{\scol\mstr\mval}'}
\\
\sdel{\srec{\vec{\scol\mstr\mval}}}{\mstr_x}
&\stepsto&
\srec{\vec{\scol\mstr\mval}}\text,
\mbox{ if }\mstr_x\not\in\vec\mstr
\\
\sif{\strue}{\mcls_1}{\mcls_2} &\stepsto& \mcls_1
\\
\sif{\sfalse}{\mcls_1}{\mcls_2} &\stepsto& \mcls_2
\\
\sseq\mval\mcls &\stepsto& \mcls
\\
\swhile{\mcls_1}{\mcls_2} & \stepsto & 
  \sif{\mcls_1}{\sseq{\mcls_2}{\swhile{\mcls_1}{\mcls_2}}}{\sundefn}
\\
\stryc\mval\mvar\mcls & \stepsto & \mval
\\
\stryf\mval\mcls & \stepsto & \sseq\mcls\mval
\\
\slab\mlab\mval & \stepsto & \mval
\\
\sopn{\mval_1\dots\mval_n} &\stepsto & \delta_n(\mop_n,\mval_1\dots\mval_n)
\\[2mm]
\multicolumn{3}{l}{\mbox{Context-sensitive, store-sensitive rules:}}
\\[1mm]
\langle\msto,\mctx[\mcls]\rangle 
&\stepsto&
\langle\msto,\mctx[\mcls']\rangle\text,
\mbox{ if }\mcls\stepsto\mcls'
\\
\langle\msto,\mctx[\sref\mval]\rangle 
&\stepsto&
\langle\msto[\maddr\mapsto\mval],\mctx[\maddr]\rangle\text,
\mbox{ where }\maddr\notin\dom(\msto)
\\
\langle\msto,\mctx[\sderef\maddr]\rangle 
&\stepsto&
\langle\msto,\mctx[\mval]\rangle\text,
\mbox{ if }\msto(\maddr) = \mval
\\
\langle\msto,\mctx[\sset\maddr\mval]\rangle 
&\stepsto&
\langle\msto[\maddr\mapsto\mval],\mctx[\mval]\rangle 
\\
\langle\msto,\mcctx[\sthrow\mval]\rangle
&\stepsto&
\langle\msto,\serr\mval\rangle
\\
\langle\msto,\mctx[\stryc{ \mcctx[\sthrow\mval] }\mvar{\scls\mexp\menv}]\rangle
&\stepsto&
\langle\msto,\mctx[\scls\mexp{\menv[\mvar\mapsto\mval]}]\rangle
\\
\langle\msto,\mctx[\stryf{ \mcctx[\sthrow\mval] }\mcls]\rangle
&\stepsto&
\langle\msto,\mctx[\sseq\mcls{\sthrow\mval}]\rangle
\\
\langle\msto,\mctx[\slab\mlab{ \mcctx[\sthrow\mval] }]\rangle
&\stepsto&
\langle\msto,\mctx[\sthrow\mval]\rangle
\\
\langle\msto,\mctx[\stryc{ \mcctx[\sbreak\mlab\mval] }\mvar\mcls]\rangle
&\stepsto&
\langle\msto,\mctx[\sbreak\mlab\mval]\rangle
\\
\langle\msto,\mctx[\stryf{ \mcctx[\sbreak\mlab\mval] }\mcls]\rangle
&\stepsto&
\langle\msto,\mctx[\sseq\mcls{\sbreak\mlab\mval}]\rangle
\\
\langle\msto,\mctx[\slab\mlab{ \mcctx[\sbreak\mlab\mval]} ]\rangle
&\stepsto&
\langle\msto,\mctx[\mval]\rangle
\\
\langle\msto,\mctx[\slab{\mlab'}{ \mcctx[\sbreak\mlab\mval] }]\rangle
&\stepsto&
\langle\msto,\mctx[\sbreak\mlab\mval]\rangle\text,
\mbox{ if }\mlab'\not=\mlab
\end{array}
\]
\caption{Reduction semantics for $\lambda\rho_\JS$}
\label{fig:reductions}
\end{figure*}

The complete syntax of values and closures in $\lambda\rho_\JS$ is
given in figure~\ref{fig:syntax-rho}.
The semantics of $\lambda\rho_\JS$ is given in terms of a small-step
reduction relation defined in figure~\ref{fig:reductions}.  There are
four classes of reductions:
\begin{enumerate}
\item context-insensitive, store-insensitive reductions operating
  over closures to implement computations that have no effect on the
  context or store,
\item context-sensitive or store-sensitive reductions operating over
  pairs of stores and programs to implement memory- and
  control-effects,
\item (omitted) reductions propagating environments from closures
  to inner expressions, and
\item (omitted) reductions raising exceptions that represent run-time
  errors such as applying a non-function, branching on a
  non-boolean\footnote{One of JavaScript's quirks are its broad
    definitions of which values act as true and false, a quirk which
    doesn't appear to be modeled here at first glance.  The desugaring
    transformation eliminates this quirk by coercing the condition in
    an {\jscode if} expression.}, indexing into a non-record or with a
  non-string key, etc.  As a result, $\lambda\rho_\JS$ programs do not
  get stuck: either they diverge or result in a value or an uncaught
  exception.
\end{enumerate}

Reduction proceeds by a program being decomposed into a redex and
evaluation context, which represents a portion of program text with a
single hole, written ``$\mhole$''.  The grammar of evaluation
contexts, defined in figure~\ref{fig:contexts}, specifies \emph{where}
in a program reduction may occur.  The notation ``$\mctx[\mcls]$''
denotes both the decomposition of a program into the evaluation
context $\mctx$ with $\mcls$ in the hole and the plugging of $\mcls$
into $\mctx$, which replaces the single hole in $\mctx$ by $\mcls$.
In addition to closures, holes may also be replaced by contexts, which
yields another context.  This is indicated with the notation
``$\mctx[\mctx']$''.

\begin{figure}
\[
\begin{grammar}
\mcctx & \produces & [\;] 
\opor \slet\mvar\mcctx\mcls
\opor \sapp\mcctx\mclss
\opor \sapp\mval{\vec\mval, \mcctx, \vec\mcls}
\opor \srec{\vec{\scol\mstr\mval}, \scol\mstr\mcctx, \vec{\scol\mstr\mcls}}
\opor \sop{\mvals,\mcctx,\mclss}
\\
&\opor & \saref\mcctx\mcls
\opor \saref\mval\mcctx
\opor \saset\mcctx\mcls\mcls
\opor \saset\mval\mcctx\mcls
\opor \saset\mval\mval\mcctx
\opor \sdel\mcctx\mcls
\opor \sdel\mval\mcctx
\\
&\opor & \sref\mcctx
\opor \sderef\mcctx
\opor \sset\mcctx\mcls
\opor \sset\mval\mcctx
\opor \sif\mcctx{\mcls}{\mcls}
\opor \sseq\mcctx\mcls
\opor \sthrow\mcctx
\opor \sbreak\mlab\mcctx
\\[1mm]
\mdctx &\produces &[\;]
\opor \stryf{\mcctx[\mdctx]}\mcls
\opor \slab\mlab{\mcctx[\mdctx]}
\opor \stryc{\mcctx[\mdctx]}\mvar\mcls
\\[1mm]
\mctx &\produces &\mcctx[\mdctx]
\end{grammar}
\]
\caption{Evaluation contexts for $\lambda\rho_\JS$}
\label{fig:contexts}
\end{figure}

There are three classes of evaluation contexts in
figure~\ref{fig:contexts}: \emph{local contexts} $\mcctx$ range over
all contexts that do not include exception handlers, finalizers, or
labels; \emph{control contexts} $\mdctx$ range over contexts that are
either empty or have a outermost exception handler, finalizer, or
label; and \emph{general contexts} $\mctx$ range over all evaluation
contexts.

The distinction is made to describe the behavior of
$\lambda\rho_\JS$'s control constructs: breaks, finalizers, and
exceptions.  When an exception is thrown, the enclosing local context
is discarded and if the nearest enclosing control context is an
exception handler, the thrown value is given to the handler:\footnote{We
omit the store from these examples since they have no effect
upon it.}
\begin{gather*}
\mctx[\stryc{\mcctx[\sthrow\mval]}\mvar{\scls\mexp\menv}]
\xstepsto
\mctx[\scls\mexp{\menv[\mvar\mapsto\mval]}]\text.
\end{gather*}
If the nearest enclosing control context is a finalizer, the
finalization clause is evaluated and the exception is rethrown:
\begin{gather*}
\mctx[\stryf{\mcctx[\sthrow\mval]}\mcls]
\xstepsto
\mctx[\sseq\mcls{\sthrow\mval}]\text.
\end{gather*}
If the nearest enclosing control context is a label, the context
up to and including the label are discarded and the thrown value
continues outward toward the next control context:
\begin{gather*}
\mctx[\slab\mlab{\mcctx[\sthrow\mval]}]
\xstepsto
\mctx[\sthrow\mval]\text.
\end{gather*}
Finally, if there is no enclosing control context, the exception
was not handled and the program has resulted in an error:
\begin{gather*}
\mcctx[\sthrow\mval] \xstepsto \serr\mval\text.
\end{gather*}

Breaks are handled in a similar way, except local contexts are
discarded until the matching label are found.  In the case of
finalizers, the finalization clause is run, followed by reinvoking the
break.

The result of a computation is an \emph{answer}, which consists of a
store and either a value or error, indicating an uncaught
exception:
\begin{gather*}
\mans \produces  \langle\msto,\mval\rangle
\opor  \langle\msto,\serr\mval\rangle
\end{gather*}

The \emph{evaluation} of a program is defined by a partial function
relating programs to answers:
\[
\evf(\expr) = \mans\text{ if } \inj_\JS(\expr) \xmultistepsto \mans\text{, for some }\mans\text,
\]
where $\xmultistepsto$ denotes the reflexive, transitive closure of
the reduction relation defined in figure~\ref{fig:reductions}
and
\[
\inj_\JS(\mexp) = \langle\mtsto,\scls\mexp\mtenv\rangle\text.
\]


Having established the syntax and semantics of the
$\lambda\rho_\JS$-calculus, we now relate it to Guha \etal's
$\lambda_\JS$-calculus.

\subsection{Correspondence with $\boldsymbol{\lambda_\JS}$}
\label{sec:correspondence}

We have developed an explicit substitution variant of $\lambda_\JS$ in
order to derive an environment-based abstract machine, which as we
will see, is important for the subsequent abstract interpretation.
However, let us briefly discuss this new calculus's relation to
$\lambda_\JS$ and establish their correspondence so that we can rest
assured that our analytic framework is really reasoning about
$\lambda_\JS$ programs.

Our presentation of evaluation contexts for $\lambda\rho_\JS$ closely
follows Guha,~\etal\ There are two important differences.
\begin{enumerate}

\item The grammar of evaluation contexts for $\lambda_\JS$ makes a
  distinction between local contexts including labels, and local
  contexts including exception handlers.  Let $\mfctx$ and $\mgctx$
  denote such contexts, respectively:
  \[
  \begin{grammar}
    \mfctx & \produces & \mcctx \opor \mcctx[\slab\mlab\mfctx]\\
    \mgctx & \produces & \mcctx \opor \mcctx[\stryc\mgctx\mvar\mcls]\text.
  \end{grammar}
  \]
  The distinction allows for exceptions to effectively jump over
  enclosing labels and for breaks to jump over handlers in one step of
  reduction:
  \begin{gather*}
    \mctx[\stryc{\mfctx[\sthrow\mval]}\mvar{\scls\mexp\menv}]    
    \xstepsto
    \mctx[\scls\mexp{\menv[\mvar\mapsto\mval]}]\text,
  \end{gather*}
  and
  \begin{align*}
    \mctx[\slab{\mlab'}{\mgctx[\sbreak\mlab\mval]}]    
    &\xstepsto \mctx[\mval] \mbox{, if }\mlab'=\mlab\\
    &\xstepsto \mctx[\sbreak\mlab\mval] \mbox{, otherwise.}
  \end{align*}

  It should be clear that our notion of reduction can simulate the
  above one-step reductions in one or more steps corresponding to the
  number of labels (exception handlers) in $\mfctx$ (in $\mgctx$).
  We adopt our single notion of label and handler free local contexts
  in order to simplify the abstract machine in the subsequent section.

\item The grammar of evaluation of contexts for $\lambda_\JS$
  mistakenly does not include break contexts in the set of local
  contexts, causing break expressions within break expression to get
  stuck, falsifying the soundness theorem.  The mistake is minor and
  easily fixed.  When relating $\lambda\rho_\JS$ to $\lambda_\JS$ we
  assume this correction has been made.
\end{enumerate}

We write \(\lambda_\JS\) and \(\lambda\rho_\JS\) over a reduction
relation to denote the (omitted) one-step reduction relation as given
by Guha,~\etal, corrected as described above, and the one-step
reduction as defined in figure~\ref{fig:reductions}, respectively.

The results of $\lambda_\JS$ and $\lambda\rho_\JS$ evaluation are
related by a function $\UNLD$ that recursively forces the all of the
delayed substitutions represented by an environment \cite[\S
  2.5]{dvanhorn:Biernacka2007Concrete}, thus mapping a value to a
syntactic value.  It is the identity function on syntactic values; for
answers, functions, and records it is defined as:
\begin{align*}
\UNLD(\langle\msto,\mval\rangle) &= \langle\msto,\UNLD(\mval)\rangle\\
\UNLD(\langle\msto,\serr\mval\rangle) &= \langle\msto,\serr{\UNLD(\mval)}\rangle\\
\UNLD\scls{\sfunc\mvars\mexp}{\{(\mvar_0,\mval_0),\dots,(\mvar_n,\mval_n)\}}
& = \sfunc\mvar{[\UNLD(\mval_0)/\mvar_0,\dots,\UNLD(\mval_n)/\mvar_n]\mexp}\\
\UNLD\scls{\srec{\vec{\scol\mstr\mval}}}{\{(\mvar_0,\mval_0),\dots,(\mvar_n,\mval_n)\}}
& = \srec{ \vec{\scol\mstr{[\UNLD(\mval_0)/\mvar_0,\dots,\UNLD(\mval_n)/\mvar_n]\mval}} }\text.
\end{align*}
We can now formally state the calculi's correspondence:
\begin{lemma}[Correspondence]
\label{lem:correspondence}
For all programs $\mexp$,
\begin{gather*}
\langle\mtsto,\mexp\rangle \xmultistepsto[\lambda_\JS] \mans \iff
\inj_\JS(\mexp)
\xmultistepsto[\lambda\rho_\JS] \mans'\text,
\end{gather*}
where $\mans = \UNLD(\mans')$.
\end{lemma}
\begin{proof}(Sketch.)  
The proof follows 
the structure of Biernacka and
Danvy's~\cite{dvanhorn:Biernacka2007Concrete} proof of correspondence
for the $\lambda$-calculus and Curien's $\lambda\rho$-calculus of
explicit substitutions~\cite{dvanhorn:Curien1991Abstract},
straightforwardly extended to $\lambda_\JS$ and $\lambda\rho_\JS$.
\end{proof}

We have now established our semantic basis: a calculus of explicit
substitutions corresponding to $\lambda_\JS$, which is a model
adequate for all of JavaScript minus {\tt eval}.  In the following
section, we apply the syntactic correspondence to derive a
correct-by-construction environment-based abstract machine.

\section{The JavaScript Abstract Machine (JAM)}
\label{sec:jam}

In the section, we present the JAM: the JavaScript Abstract Machine.
The abstract machine takes the form of a first-order state transition
system that operates over triples consisting of a store, a closure,
and a control stack, represented by a list of evaluation contexts.
There are three classes of transitions for the machine: those that
evaluate, those that continue, and those that apply.
\begin{description}
\item[evaluate:] Evaluation transitions operate over triples
  dispatching on the closure component.  The eval transitions
  implement a search for a redex or a value.  If the closure is a
  value, then the search for a value is complete; the machine
  transitions to a continue state to plug the value into its context.
  Alternatively, if the closure is a redex, then the search is also
  complete and the machine transitions to an apply state to contract
  the redex.  Finally, if the closure is a neither a redex nor a
  value, the search continues; the machine selects the next closure to
  search and pushes a single evaluation context on to the control
  stack.

\item[continue:] Continuation transitions operate over triples where
  the closure component is always a value, dispatching on the top
  evaluation context.  The value is being plugged into the context
  represented by the stack.  If plugging the value into the context
  results in a redex, the machine transitions to an apply state.  If
  plugging the value reveals the next closure that needs to be
  evaluated, the machine transitions to an evaluate state.  If
  plugging the value in turn results in \emph{another} value, the
  machine transitions to a continue state to plug that value.
  Finally, if both the control stack is empty, the result of the
  program has been reached and the machine halts with the answer.

\item[apply:] Application transitions operate over triples where the
  closure component is always a redex.  These transitions dispatch
  (mostly) on the redex and serve to contract it, thus implementing
  the reduction relation of figure~\ref{fig:reductions}.  Since
  reductions are potentially store- and context-sensitive, the
  transitions may also dispatch on the control continuation in order
  to implement the control operators.
\end{description}

\begin{figure}[t]
\[
\begin{array}{lcl}
\langle\msto,(\mvar,\menv),\mETX\erangle
&\xstepsto&
\langle\msto,(\mvar,\menv),\mETX\arangle
\\
\langle\msto,\mval,\mETX\erangle
&\xstepsto&
\langle\msto,\mval,\mETX\crangle
\\
\langle\msto,\srec{\scol\mstr\mcls,\dots},\mETX\erangle
&\xstepsto&
\langle\msto,\mcls,\cons{\srec{\scol\mstr\mhole,\dots}}\mETX\erangle
\\
\langle\msto,\slet\mvar{\mcls}{\mclso},\mETX\erangle
&\xstepsto&
\langle\msto,\mcls,\cons{\slet\mvar\mhole\mclso}\mETX\erangle
\\
\langle\msto,\sapp\mcls{\vec\mclso},\mETX\erangle
&\xstepsto&
\langle\msto,\mcls,\cons{\sapp\mhole{\vec\mcls}}\mETX\erangle
\\
\langle\msto,\saref{\mcls}{\mclso},\mETX\erangle
&\xstepsto&
\langle\msto,\mcls,\cons{\saref\mhole\mclso}\mETX\erangle
\\
\langle\msto,\saset\mcls\mclso{\mclso'},\mETX\erangle
&\xstepsto&
\langle\msto,\mcls,\cons{\saset\mhole\mclso{\mclso'}}\mETX\erangle
\\
\langle\msto,\sdel\mcls\mclso,\mETX\erangle
&\xstepsto&
\langle\msto,\mcls,\cons{\sdel\mhole{\mclso}}\mETX\erangle
\\
\langle\msto,\sset\mcls\mclso,\mETX\erangle
&\xstepsto&
\langle\msto,\mcls,\cons{\sset\mhole\mclso}\mETX\erangle
\\
\langle\msto,\sref\mcls,\mETX\erangle
&\xstepsto&
\langle\msto,\mcls,\cons{\sref\mhole}\mETX\erangle
\\
\langle\msto,\sderef\mcls,\mETX\erangle
&\xstepsto&
\langle\msto,\mcls,\cons{\sderef\mhole}\mETX\erangle
\\
\langle\msto,\sif{\mcls}{\mclso}{\mclso'},\mETX\erangle
&\xstepsto&
\langle\msto,\mcls,\cons{\sif\mhole\mclso{\mclso'}}\mETX\erangle
\\
\langle\msto,\sseq{\mcls}{\mclso},\mETX\erangle
&\xstepsto&
\langle\msto,\mcls,\cons{\sseq\mhole\mclso}\mETX\erangle
\\
\langle\msto,\swhile{\mcls}{\mclso},\mETX\erangle
&\xstepsto&
\langle\msto,\mcls, \cons{\sif\mhole{\sseq\mclso{\swhile\mcls\mclso}}\sundefn}\mETX \erangle
\\
\langle\msto,\slab\mlab\mcls,\mETX\erangle
&\xstepsto&
\langle\msto,\mcls,\cons{\slab\mlab\mhole}\mETX\erangle
\\
\langle\msto,\sbreak\mlab\mcls,\mETX\erangle
&\xstepsto&
\langle\msto,\mcls,\sbreak\mlab\mhole::\mETX \erangle
\\
\langle\msto,\stryc{\mcls}\mvar{\mclso},\mETX\erangle
&\xstepsto&
\langle\msto,\mcls,\cons{\stryc{\mhole}\mvar{\mclso}}\mETX\erangle
\\
\langle\msto,\stryf{\mcls}{\mclso},\mETX\erangle
&\xstepsto&
\langle\msto,\mcls,\cons{\stryf{\mhole}{\mclso}}\mETX\erangle
\\
\langle\msto,\sthrow\mcls,\mETX\erangle
&\xstepsto&
\langle\msto,\mcls,\cons{\sthrow\mhole}\mETX\erangle
\\
\langle\msto,\sop{\mcls,\dots},\mETX\erangle
&\xstepsto&
\langle\msto,\mcls,\cons{\sop{\mhole,\dots}}\mETX\erangle
\end{array}
\]
\caption{Evaluation transitions}
\end{figure}

\begin{figure}[t]
\[
\begin{array}{lcl}
\langle\msto,\mval,\nil\crangle
&\xstepsto&
\langle\mval,\msto\rangle
\\
\langle\msto,\mval,\cons{\slet\mvar\mhole\mcls}\mETX\crangle
&\xstepsto&
\langle\msto,\slet\mvar\mval\mcls,\mETX\arangle
\\
\langle\msto,\mval,\cons{\sapp\mhole\relax}\mETX\crangle
&\xstepsto&
\langle\msto,\sapp\mval\relax,\mETX\arangle
\\
\langle\msto,\mval,\cons{\sapp\mhole{\mcls,\dots}}\mETX\crangle
&\xstepsto&
\langle\msto,\mcls,\cons{\sapp\mval{\mhole,\dots}}\mETX\erangle
\\
\langle\msto,\mval,\cons{\sapp\mvaloo{\mvalo,\dots,\mhole}}\mETX\crangle
&\xstepsto&
\langle\msto,\sapp\mvaloo{\mvalo,\dots,\mval},\mETX\arangle
\\
\langle\msto,\mval,\cons{\sapp\mvaloo{\mvalo,\dots,\mhole,\mcls,\dots}}\mETX\crangle
&\xstepsto&
\langle\msto,\mcls,\cons{\sapp\mvaloo{\mvalo,\dots,\mval,\mhole,\dots}}\mETX\erangle
\\
\langle\msto,\mval,\cons{\srec{\scol{\mstr_1}\mvalo,\dots,\scol{\mstr_n}\mhole}}\mETX\crangle
&\xstepsto&
\langle\msto,\srec{\scol{\mstr_1}\mvalo,\dots,\scol{\mstr_n}\mval},\mETX\crangle
\\
\langle\msto,\mval,\cons{\srec{\scol{\mstr_1}\mvalo,\dots,\scol{\mstr_i}\mhole,\scol{\mstr_{i+1}}\mcls,\dots}}\mETX\crangle
&\xstepsto&
\langle\msto,\mcls,\cons{\srec{\scol{\mstr_1}\mvalo,\dots,\scol{\mstr_i}\mval,\scol{\mstr_{i+1}}\mhole,\dots}}\mETX\erangle
\\
\langle\msto,\mval,\cons{\saref\mhole\mcls}\mETX\crangle
&\xstepsto&
\langle\msto,\mcls,\cons{\saref\mval\mhole}\mETX\erangle
\\
\langle\msto,\mval,\cons{\saref\mvalo\mhole}\mETX\crangle
&\xstepsto&
\langle\msto,\saref\mvalo\mval,\mETX\arangle
\\
\langle\msto,\mval,\cons{\saset\mhole\mcls\mclso}\mETX\crangle
&\xstepsto&
\langle\msto,\mcls,\cons{\saset\mval\mhole\mclso}\mETX\erangle
\\
\langle\msto,\mval,\cons{\saset\mvalo\mhole\mcls}\mETX\crangle
&\xstepsto&
\langle\msto,\mcls,\cons{\saset\mvalo\mval\mhole}\mETX\erangle
\\
\langle\msto,\mval,\cons{\saset\mvalo\mvaloo\mhole}\mETX\crangle
&\xstepsto&
\langle\msto,\saset\mvalo\mvaloo\mval,\mETX\arangle
\\
\langle\msto,\mval,\cons{\sdel\mhole\mcls}\mETX\crangle
&\xstepsto&
\langle\msto,\mcls,\cons{\sdel\mval\mhole}\mETX\erangle
\\
\langle\msto,\mval,\cons{\sdel\mvalo\mhole}\mETX\crangle
&\xstepsto&
\langle\msto,\sdel\mvalo\mval,\mETX\arangle
\\
\langle\msto,\mval,\cons{\sref\mhole}\mETX\crangle
&\xstepsto&
\langle\msto,\sref\mval,\mETX\arangle
\\
\langle\msto,\mval,\cons{\sderef\mhole}\mETX\crangle
&\xstepsto&
\langle\msto,\sderef\mval,\mETX\arangle
\\
\langle\msto,\mval,\cons{\sset\mhole\mcls}\mETX\crangle
&\xstepsto&
\langle\msto,\mcls,\cons{\sset\mval\mhole}\mETX\erangle
\\
\langle\msto,\mval,\cons{\sset\mvalo\mhole}\mETX\crangle
&\xstepsto&
\langle\msto,\sset\mvalo\mval,\mETX\arangle
\\
\langle\msto,\mval,\cons{\sif\mhole{\mcls}{\mclso}}\mETX\crangle
&\xstepsto&
\langle\msto,\sif{\mval}{\mcls}{\mclso},\mETX\arangle
\\
\langle\msto,\mval,\cons{\sseq\mhole\mcls}\mETX\crangle
&\xstepsto&
\langle\msto,\mcls,\mETX\erangle
\\
\langle\msto,\mval,\cons{\sthrow\mhole}\mETX\crangle
&\xstepsto&
\langle\msto,\sthrow\mval,\mETX\arangle
\\
\langle\msto,\mval,\cons{\sbreak\mlab\mhole}\mETX\crangle
&\xstepsto&
\langle\msto,\sbreak\mlab\mval,\mETX\arangle
\\
\langle\msto,\mval,\cons{\sop{\mvalo,\dots,\mhole}}\mETX\crangle
&\xstepsto&
\langle\msto,\sop{\mvalo,\dots,\mval},\mETX\arangle
\\
\langle\msto,\mval,\cons{\sop{\mvalo,\dots,\mhole,\mcls,\dots}}\mETX\crangle
&\xstepsto&
\langle\msto,\mcls,\cons{\sop{\mvalo,\dots,\mval,\mhole,\dots}}\mETX\erangle
\end{array}
\]
\caption{Continuation transitions}
\end{figure}

\begin{figure}[t]
\[
\begin{array}{lcl}
\langle\msto,(\mvar,\menv),\mETX\arangle
&\xstepsto&
\langle\msto,\mval,\mETX\crangle\mbox{ if }\mval\in\getf\msto\maddr 
\\
\langle\msto,\slet\mvar\mval\mcls,\mETX\arangle
&\xstepsto&
\langle\putf\msto\maddr\mval,\scls\mexp{\menv[\mvar\mapsto\maddr]},\mETX\erangle
\\
&&\mbox{where }\maddr=\allocf\state
\\
\langle\msto,\sapp{\scls{\sfunc\mvars\mexp}\menv}\mvals,\mETX\arangle
&\xstepsto&
\langle\putf\msto\maddrs\mvals,(\mexp,\menv[\mvars\mapsto\maddrs]),\mETX\erangle
\\
&&\mbox{if }|\mvars|=|\mvals|, \mbox{ where }\maddrs=\allocf\state
\\
\langle\msto,\saref{\srec{\vec{\scol\mstr\mval},\scol{\mstr_i}\mval,\vec{\scol\mstr\mval}'}}{\mstr_i},\mETX\arangle
&\xstepsto&
\langle\msto,\mval,\mETX\crangle
\\
\langle\msto,\saref{\srec{\vec{\scol\mstr\mval}}}{\mstr_x},\mETX\arangle
&\xstepsto&
\langle\msto,\sundefn,\mETX\crangle\mbox{ if }\mstr_x\notin\vec\mstr
\\
\langle\msto,\saset{\srec{\vec{\scol\mstr\mval},\scol{\mstr_i}{\mval_i},\vec{\scol\mstr\mval}'}}{\mstr_i}\mval,\mETX\arangle
&\xstepsto&
\langle\msto,\srec{\vec{\scol\mstr\mval},\scol{\mstr_i}{\mval},\vec{\scol\mstr\mval}'},\mETX\crangle
\\
\langle\msto,\sdel{\srec{\vec{\scol\mstr\mval},\scol{\mstr_i}{\mval_i},\vec{\scol\mstr\mval}'}}{\mstr_i},\mETX\arangle
&\xstepsto&
\langle\msto,\srec{\vec{\scol\mstr\mval},\vec{\scol\mstr\mval}'},\mETX\crangle
\\
\langle\msto,\sdel{\srec{\vec{\scol\mstr\mval}}}{\mstr_x},\mETX\arangle
&\xstepsto&
\langle\msto,\srec{\vec{\scol\mstr\mval}},\mETX\crangle\mbox{ if }\mstr_x\notin\vec\mstr
\\
\langle\msto,\sif\strue\mcls\mclso,\mETX\arangle
&\xstepsto&
\langle\msto,\mcls,\mETX\erangle
\\
\langle\msto,\sif\sfalse\mcls\mclso,\mETX\arangle
&\xstepsto&
\langle\msto,\mclso,\mETX\erangle
\\
\langle\msto,\sopn{\mval_1,\dots,\mval_n},\mETX\crangle
&\xstepsto&
\langle\msto,\mval,\mETX\crangle\mbox{ if }\delta(\mop_n,\mval_1,\dots,\mval_n)=\mval
\\
\langle\msto,\sref\mval,\mETX\arangle
&\xstepsto&
\langle\putf\msto\maddr\mval,\maddr,\mETX\crangle
\\
&&\mbox{where }\maddr=\allocf\state
\\
\langle\msto,\sderef\maddr,\mETX\arangle
&\xstepsto&
\langle\msto,\mval,\mETX\crangle\mbox{ if }\mval\in\getf\msto\maddr 
\\
\langle\msto,\sset\maddr\mval,\mETX\arangle
&\xstepsto&
\langle\putf\msto\maddr\mval,\mval,\mETX\crangle
\\
\langle\msto,\sthrow\mval,\nil\rangle
&\xstepsto&
\langle \serr\mval,\msto\rangle
\\
\langle\msto,\sthrow\mval,\cons{\stryc\mhole\mvar{\scls\mexp\menv}}\mETX\arangle
&\xstepsto&
\langle\putf\msto\maddr\mval,\scls\mexp{\menv[\mvar\mapsto\maddr]},\mETX\erangle
\\
&&\mbox{where }\maddr=\allocf\state
\\
\langle\msto,\sthrow\mval,\cons{\stryf\mhole\mcls}\mETX\arangle
&\xstepsto&
\langle\msto,\sseq\mcls{\sthrow\mval},\mETX\erangle
\\
\langle\msto,\sthrow\mval,\cons{\slab\mlab\mhole}\mETX\arangle
&\xstepsto&
\langle\msto,\sthrow\mval,\mETX\arangle
\\
\langle\msto,\sthrow\mval,\cons\mcctx\mETX\arangle
&\xstepsto&
\langle\msto,\sthrow\mval,\mETX\arangle 
\\
\langle\msto,\sbreak\mlab\mval,\cons{\stryc\mvar\mhole\mcls}\mETX\arangle
&\xstepsto&
\langle\msto,\sbreak\mlab\mval,\mETX\erangle
\\
\langle\msto,\sbreak\mlab\mval,\cons{\stryf\mhole\mcls}\mETX\arangle
&\xstepsto&
\langle\msto,\sseq\mcls{\sbreak\mlab\mval},\mETX\erangle
\\
\langle\msto,\sbreak\mlab\mval,\cons{\slab\mlab\mhole}\mETX \arangle
&\xstepsto&
\langle\msto,\mval,\mETX\crangle
\\
\langle\msto,\sbreak\mlab\mval,\cons{\slab{\mlab'}\mhole}\mETX \arangle
&\xstepsto&
\langle\msto,\mval,\mETX\crangle\mbox{ if }\mlab\neq\mlab'
\\
\langle\msto,\sbreak\mlab\mval,\cons\mcctx\mETX\arangle
&\xstepsto&
\langle\msto,\sbreak\mlab\mval,\mETX\arangle
\end{array}
\]
\caption{Application transitions}
\end{figure}

The machine relies on three functions for interacting with the store:
\begin{align*}
\allocff &: \s{State} \rightarrow \s{Address}^n\\
\putff &: \s{Store} \times \s{Address} \times \s{Value} \rightarrow \s{Store}\\
\getff &: \s{Store} \times \s{Address} \rightarrow \mathcal{P}(\s{Value})
\end{align*}
The $\allocff$ function makes explicit the non-deterministic choice of
addresses when allocating space in the store.  It returns a vector of
addresses, often a singleton, based on the current state of the
machine.  For the moment, all that we require of $\allocff$ is that it
return a suitable number of addresses and that none are in use in the
store.  The $\putff$ function updates a store location and is defined
as:
\begin{align*}
\putf\msto\maddrs\mvals = \msto[\maddrs\mapsto\mvals]\text,
\end{align*}
and the $\getff$ function retrieves a value from a store location as a
singleton set:
\begin{align*}
\getf\msto\maddr = \{\msto(\maddr)\}\text.
\end{align*}
We make explicit the use of these three functions because they will
form the essential mechanism of abstracting the machine in the
subsequent section; the slightly strange definition for $\getff$ is to
facilitate approximation where there may be multiple values residing
at a store location.

The initial machine configuration is an evaluation state consisting of
the empty store, the program, the empty environment, and the empty
control and local continuation:
\begin{gather*}
\inj_\JAM(\mexp) = \langle\mtsto,\scls\mexp\mtenv,\nil\erangle\text.
\end{gather*}
Final configurations are answers, just as in the reduction semantics.

\subsection{Reformulation of reduction semantics}

Unfortunately, there is an immediate problem with the described
approach when applied to the JavaScript abstract machine.  The problem
stems from the JAM having two control stacks.  Consequently, when
abstracting we arrive at a two-stack pushdown machine, which in
general has the power to simulate a Turing-machine.  However this
problem can be overcome: the JAM can be reformulated into a single
stack machine in such a way that preserves correctness and enables a
pushdown abstraction that is decidable.

One of the lessons of our abstract machine-based approach to analysis
is that many problems in program analysis can be solved at the
semantic level and then imported systematically to the analytic side.
So for example, abstract garbage
collection~\cite{mattmight:Might:2008:Exploiting} can be expressed as
concrete garbage collection with the pointer refinement and store
abstraction applied~\cite{dvanhorn:VanHorn2010Abstracting}.
Similarly, the exponential complexity of $k$-CFA can be avoided by
concretely changing the representation of closures and then
abstracting in an unremarkable way~\cite{dvanhorn:Might2010Resolving}.

We likewise solve our two-stack problem by a reformulation at the
level of the reduction semantics for $\lambda\rho_\JS$ and then repeat
the refocusing construction to derive a one-stack variant of the JAM.

\begin{figure}
\[
\begin{array}{rcl}
\langle\msto,\sthrow\mval\rangle
&\stepsto&
\langle\msto,\serr\mval\rangle
\\
\langle\msto,\mctx[\msctx[\sthrow\mval]]\rangle
&\stepsto&
\langle\msto,\mctx[\sthrow\mval]\rangle
\\
\langle\msto,\mctx[\stryc{ \sthrow\mval }\mvar{\scls\mexp\menv}]\rangle
&\stepsto&
\langle\msto,\mctx[\scls\mexp{\menv[\mvar\mapsto\mval]}]\rangle
\\
\langle\msto,\mctx[\stryf{ \sthrow\mval }\mcls]\rangle
&\stepsto&
\langle\msto,\mctx[\sseq\mcls{\sthrow\mval}]\rangle
\\
\langle\msto,\mctx[\slab\mlab{ \sthrow\mval }]\rangle
&\stepsto&
\langle\msto,\mctx[\sthrow\mval]\rangle
\\
\langle\msto,\mctx[\msctx[\sbreak\mlab\mval]]\rangle
&\stepsto&
\langle\msto,\mctx[\sbreak\mlab\mval]\rangle
\\
\langle\msto,\mctx[\stryc{ \sbreak\mlab\mval }\mvar\mcls]\rangle
&\stepsto&
\langle\msto,\mctx[\sbreak\mlab\mval]\rangle
\\
\langle\msto,\mctx[\stryf{ \sbreak\mlab\mval }\mcls]\rangle
&\stepsto&
\langle\msto,\mctx[\sseq\mcls{\sbreak\mlab\mval}]\rangle
\\
\langle\msto,\mctx[\slab\mlab{ \sbreak\mlab\mval } ]\rangle
&\stepsto&
\langle\msto,\mctx[\mval]\rangle
\\
\langle\msto,\mctx[\slab{\mlab'}{ \sbreak\mlab\mval }]\rangle
&\stepsto&
\langle\msto,\mctx[\sbreak\mlab\mval]\rangle\text,
\mbox{ if }\mlab'\not=\mlab
\end{array}
\]
\caption{Reformulated reduction semantics}
\label{fig:reductions-reformulated}
\end{figure}

The basic reason for maintaining the control and local stack is to
allow jumps over the local context whenever a control operator is
invoked.  This is seen in the reduction semantics with reductions such
as this:
\begin{align*}
\mctx[\slab{\mlab'}{\mcctx[\sbreak\mlab\mval]}]
&\xstepsto \mctx[\mval] \mbox{ if }\mlab'=\mlab\\
&\xstepsto \mctx[\sbreak\mlab\mval] \mbox{ if }\mlab'\not=\mlab
\end{align*}
To enable a single stack, we simulate this jump over the local context
by ``bubbling'' over it in a piecemeal fashion.  This is accomplished
by defining a notion of single, non-empty local context frames
$\msctx$, \ie,
\[
\begin{grammar}
\msctx &\produces& \slet\mvar\mhole\mcls
\opor \sapp\mhole\mclss
\opor \sapp\mcls{\mvals,\mhole,\mclss}
\opor \cdots
\opor \sbreak\mlab\mhole
\opor \sop{\mvals,\mhole,\mclss}
\end{grammar}
\]
then the reduction relation for control operators remains context
sensitive, but does not operate over whole contexts, but just the
enclosing frame, which can then be implemented with a stack.  The
rules for simulating the above reduction are then:
\begin{align*}
\mctx[\msctx[\sbreak\mlab\mval]] &\xstepsto \mctx[\sbreak\mlab\mval]
\\
\mctx[\slab{\mlab'}{\sbreak\mlab\mval}]
&\xstepsto
\mctx[\mval]\mbox{ if }\mlab'=\mlab
\\
& \xstepsto
\mctx[\sbreak\mlab\mval]\mbox{ if }\mlab'\not=\mlab
\end{align*}
Clearly, these reductions simulate the original single reduction by a
number of steps that corresponds to the number of local frames between
the label and the break.

The complete replacement of the context-sensitive reductions is given
in figure~\ref{fig:reductions-reformulated}.  We refer to this
alternative reduction semantics as $\lambda\rho'_\JS$.

\begin{lemma}
For all programs $\mexp$,
\begin{gather*}
\inj_\JS(\mexp)
\xmultistepsto[\lambda\rho_\JS] \mans \iff
\inj_\JS(\mexp)
\xmultistepsto[\lambda\rho'_\JS] \mans\text.
\end{gather*}

\end{lemma}

Guha~\etal{} handle primitive operations in $\lambda_{JS}$ in standard fashion
by delegating to a $\delta$-function.
For the sake of analysis, we can delegate to any sound, finite abstraction of
the $\delta$-function.
The simplest such abstraction maps values to their types, which makes the
abstract $\delta$ function isomorphic to its intensional signature.
For in-depth discussion of  richer  abstract domains over basic values for use 
in JavaScript, we refer the reader to Jensen~\etal~\cite{dvanhorn:Jensen2009Type};
they provide abstract domains for JavaScript which could be plugged directly into the $\AJAM$.

\subsection{Correctness of the JAM}

The JAM is a correct evaluator for $\lambda\rho_\JS$, and hence for
$\lambda_\JS$ as well.
\begin{lemma}[Correctness]
\label{lem:correct}
For all programs $\mexp$,
\begin{gather*}
\inj_\JS(\mexp)\xmultistepsto[\lambda_\JS]\mans
\iff
\inj_\JAM(\mexp) \xmultistepsto[\JAM]\mans
\text.
\end{gather*}
\end{lemma}
\begin{proof}(Sketch.)  
%
The correctness of the machine follows from the correctness of
refocusing~\cite{dvanhorn:Danvy-Nielsen:RS-04-26} and the (trivial)
meaning preservation of subsequent transformations.

The detailed step-by-step transformation from the reduction semantics
to the abstract machine has been carried out in the meta-language of
SML.
\end{proof}

For the purposes of program analysis, we rely on the following
definition of a program's reachable machine states, where $\state$
ranges over states:
\[
\JAM(\mexp) = \{ \state\ |\ \inj_\JAM(\mexp) \xmultistepsto[\JAM] \state \}\text.
\]

The set \(\JAM(\mexp)\) is potentially infinite and membership is
clearly undecidable.  In the next section, we devise a sound and
computable approximation to the set \(\JAM(\mexp)\) by a family of
pushdown automata.

\section{Pushdown abstractions of JavaScript}
\label{sec:pda}

To model non-local control precisely, the analysis must model the
program stack precisely.  Yet, the program stack can grow without
bound---a substantial obstacle for the finite-state framework.
Pushdown abstraction maps that program stack onto the unbounded stack
of a pushdown system.  Because the analysis inflicts a finite-state
abstraction on the \emph{control states} of the pushdown system, the
analysis remains decidable.

The idea is that by bounding the store, the control stack, while
unbounded, will consist of a finite stack alphabet.  Since the
remaining components of the machine are finite, the abstract machine
is equivalent in computational power to a pushdown automaton, and thus
reachability questions are casts naturally in terms of decidable PDA
reachability properties.


\subsection{Bounding the store, not the stack}

\[
\AJAM(\mexp) = \{ \astate\ |\ \inj_\JAM(\mexp) \xmultistepsto[\AJAM] \astate \}\text.
\]

The abstracted JAM provides a sound simulation of the JAM and, by
lemmas~\ref{lem:correspondence} and~\ref{lem:correct}, a sound
simulation of $\lambda\rho_\JS$, and $\lambda_\JS$, as well.

The machine's state-space is bounded simply by restricting the set of
allocatable addresses to a fixed set of finite size, $\sa{Address}$.
This necessitates a change in the machine transition system and the
representation of states.  The machine can no longer restrict
allocated addresses to be fresh with respect to the domain of the
store as is the case when bindings are allocated, ref-expression are
evaluated, and continuations are pushed.  Instead, the machine calls
an allocation function that returns a member of the finite set of
addresses.  Since the allocation function may return an address
already in use, the behavior of the store must change to accommodate
multiple values residing in a given location.  We let $\astore$ range
over such stores:
\[
\astore \in \sa{Store} = \sa{Address} \rightarrow_{\text{fin}} \mathcal{P}(\s{Value})\text.
\]
We let $\AJAM$ denote the abstract machine that results from
replacing all occurrences of the functions $\allocff$, $\putff$,
and $\getff$, with the following counterparts:
\begin{align*}
\aallocff &: \s{State} \rightarrow \sa{Address}\;\!^n
\\
\aputff &: \sa{Store} \times \sa{Address}\;\!^n \times \s{Value}\;\!^n \rightarrow \sa{Store}
\\
\agetff &: \sa{Store} \times \sa{Address} \rightarrow \mathcal{P}(\s{Value})
\end{align*}
The $\aallocff$ function works like $\allocff$, but produces addresses
from the finite set $\sa{Address}$.  The $\aputff$ function updates a
store location by \emph{joining} the given value to any existing
values that reside at that address:
\[
\aputf\masto\maddr\mval = \masto[\maddr\mapsto\{\mval\}\cup\masto(\maddr)]\text.
\]
Joining rather than updating is critical for maintaining soundness.

In essence, the finiteness of the address space implies collisions may
occur in the store.  By joining, we ensure these collisions are
modelled safely.  The $\agetff$ function returns the set of values at
a store location, allowing a non-deterministic choice of values at
that location.

We can formally relate the JAM to its abstracted counterpart through
the natural structural abstraction map $\absmap$ on their
state-spaces.
This map recurs over the state-space of the JAM to inflict a
finitizing abstraction at the leaves its state-space---addresses and
primitive values---and structures that cannot soundly absorb that
finitization, which in this case, is only the store.
The range of the store expands into a power set, so that when an abstract
address is \emph{re}-allocated, it can hold both the existing values and the
newly added value; formally:
\begin{equation*}
\absmap(\store) = \lambda \aaddr . \!\! \bigsqcup_{\absmap(\addr) = \aaddr} \!\! \absmap(\store(\addr))\text.
\end{equation*}

\begin{theorem}[Soundness]
  If $\state \xstepsto[\JAM] \state'$ and $\alpha(\state) \wt
  \astate$,
then there exists an abstract state $\astate'$, such that
  $\astate \xstepsto[\AJAM] \astate'$ and $\alpha(\state') \wt
  \astate'$.
\end{theorem}
\begin{proof}
  We reason by case analysis on the transition.
  In the cases where the transition is deterministic, the result
  follows by calculation.
  For the the remaining non-deterministic cases, we must show
  an abstract state exists such that the simulation is
  preserved.
  By examining the rules for these cases, we see that all hinge on the
  abstract store in $\astate$ soundly approximating the concrete store
  in $\state$, which follows from the assumption that $\absmap(\state)
  \wt \astate$.
\end{proof}

The more interesting aspect of the pushdown abstraction is
decidability.  Notice that since the stack has a recursive, unbounded
structure, the state-space of the machine is potentially infinite so
deciding reachability by enumerating the reachable states will no
longer suffice.

\begin{theorem}[Decidability]
$\astate \in \AJAM(\mexp)$ is decidable.
\end{theorem}
\begin{proof}
  States of the abstracted JAM consist of a store, a closure, and a
  list of single evaluation contexts representing the control stack.
  Observe that with the exception of the stack, each of these sets is
  finite: for a given program, there are a fixed set of expressions;
  environments are finite since they map variables to addresses and
  the address space is bounded; since expressions and environments are
  finite, so too are the set of values; stores are finite since
  addresses and values are finite.

  For the machine transitions that dispatch on the control stack, only
  the top-most element is used and stack operations always either push
  or pop a single context frame on at a time,~\ie, the machine obeys a
  stack discipline.  The stack alphabet consists of single evaluation
  contexts, which include a number of expressions, a value, or an
  environment, all of which are finite sets.  Thus the stack alphabet
  is finite. Consequently the machine is a pushdown automaton and
  decidability follows from known results on pushdown automata.
\end{proof}

\subsection{Instantiations}

We have now described the design of a sound and decidable framework
for the pushdown analysis of JavaScript.  The framework has a single
point of control for governing the precision of an analysis, namely
the $\aalloc$ function.  The restrictions on acceptable $\aalloc$
functions are fairly liberal: \emph{any} allocation policy is sound,
and so long as the policy draws from a finite set of addresses for a
given program, the analysis will be decidable.

At its simplest, the $\aalloc$ function could produce a constant
address:
\[
\aallocf\state = \text{\tt a}\text.
\]

A more refined analysis can be obtained by a more refined allocation
policy.  0CFA for example, distinguishes bindings of differently named
variables, but merges all bindings for a given variable name.  The
allocation function corresponding to this strategy is:
\begin{align*}
\aallocf{\langle\msto,\slet\mvar\mval\mcls,\mETX\arangle} &= \mvar\\
\aallocf{\langle\msto,\sapp{\scls{\sfunc\mvars\mexp}\menv}\mvals,\mETX\arangle} &= \mvars\\
\aallocf{\langle\msto,\sthrow\mval,\cons{\stryc\mhole\mvar\mcls}\mETX\arangle} &= \mvar
\end{align*}
This strategy uses variables names as addresses and always allocates
the variable names being bound.  The strategy is finite for a given
program and produces a pushdown generalization of classical 0CFA.
Moreover, the state-space simplifies greatly since it can be observed
that under this strategy, every environment is the identity function
on variable names.  Thus environments could be eliminated from the
semantics.

There is still the need to designin a heap abstraction, \ie, what
should the allocation function produce for:
\[
\aallocf{\langle\msto,\sref\mval,\mETX\arangle}\text?
\]
Shivers' original formulation of 0CFA had a very simple heap
abstraction corresponding to the constant allocation function
above~\cite{dvanhorn:Shivers:1991:CFA}.  More refined heap
abstractions are obtained by simply designing better strategies for
this case of $\aalloc$.

The $k$-CFA hierarchy, of which 0CFA is the base, refines the above
allocation policy by pairing variables together with bounded history
of the calling context at the binding site of a variable.  Such an
abstraction is easily expressible in our framework as follows:
\begin{align*}
\aallocf{\langle\msto,\slet\mvar\mval\mcls,\mETX\arangle} &= \langle\mvar,\lfloor\mETX\rfloor_k\rangle\\
\aallocf{\langle\msto,\sapp{\scls{\sfunc\mvars\mexp}\menv}\mvals,\mETX\arangle} &= \langle\mvars,\lfloor\mETX\rfloor_k\rangle\\
\aallocf{\langle\msto,\sthrow\mval,\cons{\stryc\mhole\mvar\mcls}\mETX\arangle} &= \langle\mvar,\lfloor\mETX\rfloor_k\rangle\text,
\end{align*}
where
\begin{align*}
\lfloor\nil\rfloor_k &= \nil\\
\lfloor\mETX\rfloor_0 & = \nil\\
\lfloor\cons\mectx\mETX\rfloor_{k+1} & = \cons\mectx{\lfloor\mETX\rfloor_k}\text.
\end{align*}
This strategy uses variable names paired together with a fixed depth
view of the control stack to allocate bindings.  It is easy to vary
this strategy for various $k$-CFA like abstraction, \eg, taking the
top $k$ \emph{different} stack frames, or taking the top $k$
application frames by filtering out non-application frames.

By giving alternative definitions of $\aalloc$ it is straightforward
to design pushdown versions of other known analyses such as
CPA~\cite{dvanhorn:Agesen1995Cartesian},
sub-0CFA~\cite{dvanhorn:ashley-dybvig-toplas98}, and
$m$-CFA~\cite{mattmight:Might:2010:mCFA}.


\subsection{Implementation}

To empirically substantiate our formal claims, we have developed
executable models and test beds.  We have developed the reduction
semantics of $\lambda\rho_\JS$ in Standard ML (SML) and carried out
the refocusing construction and subsequent program transformations in
a step-by-step manner closely following the lecture notes of
Danvy~\cite{dvanhorn:Danvy:AFP08}.  We found using SML as a
metalanguage helpful due to its type system and non-exhaustive and
redundant pattern matching warnings.  For example, we were able to
encode Guha~\etal's soundness theorem, which is false without the
modification to the semantics as described in
section~\ref{sec:correspondence}, in SML in such a way that the type
of the one-step reduction relation, coupled with exhaustive pattern
matching, implies a program is either a value or can make progress.

We ported our semantics and concrete machines to PLT
Redex~\cite{dvanhorn:Felleisen2009Semantics} and then built their
abstractions.  This was done because PLT Redex supports programming
with relations and includes a property-based random testing mechanism.
The support for programming with relations is an important aspect for
building the non-deterministic transition systems of the abstracted
JAM machines since, unlike their concrete counterparts, the transition
system cannot be encoded as a function in a straightforward way.
Using the random testing framework~\cite{dvanhorn:Klein2010Random}, we
tested the correspondence, correctness, and soundness theorems.  As an
added benefit, we were able to visualize our test programs'
state-spaces using the included graphical tools.

Finally, we used Guha~\etal's code for desugaring in order to test our
framework on real JavaScript code.  We tested against the same test
bed as Guha~\etal: a significant portion of the Mozilla JavaScript
test suite; about 5,000 lines of unmodified code.  We tested the
closure-based semantics of $\lambda\rho_\JS$ for correspondence
against the substitution-based semantics of $\lambda_\JS$ and tested
the machines for correctness with respect to the $\lambda\rho_\JS$
semantics.  Finally, we tested the instantiations of our analytic
framework for soundness with respect to the machines.  Since the
semantics of $\lambda_\JS$ have been validated against the output of
Rhino, V8, and SpiderMonkey, and all of semantic artifacts simulate or
approximate $\lambda_\JS$, these tests substantiate our framework's
correctness.

\section{Related work}
\label{sec:related}

Our approach fits cleanly within the progression of work in abstract
interpretation~\cite{mattmight:Cousot:1977:AI,mattmight:Cousot:1979:Galois}
and is inspired by the pioneering work on higher-order program
analysis by Jones~\cite{dvanhorn:Jones:1981:LambdaFlow}.  Like Jones,
our work centers around machine-based abstractions of higher-order
languages; and like Jones~\cite{dvanhorn:Schmidt2007Statetransition},
we have obtained our machines by program transformations of high-level
semantic descriptions in the tradition of
Reynolds~\cite{dvanhorn:reynolds-hosc98}.
We have been able to leverage the refocusing approach of Danvy,~\etal,
to systematically derive such
machines~\cite{dvanhorn:Danvy-Nielsen:RS-04-26,dvanhorn:Biernacka2007Concrete,dvanhorn:Danvy:AFP08},
and our main technical insight has been that threading bindings---but
not continuations---through the store results in straightforward and
clearly sound framework that precisely reasons about control flow in
face of higher-order functions and sophisticated control operators.

\subsection{Pushdown analyses}

The most closely related work to ours is Vardoulakis and Shivers
recent work on CFA2~\cite{dvanhorn:Vardoulakis2011CFA2}.
CFA2 is a table-driven summarization algorithm that exploits the
balanced nature of calls and returns to improve return-flow precision
in a control-flow analysis for CPS programs.
Though CFA2 alludes to exploiting context-free languages, context-free
languages are not explicit in its formulation in the same way that
pushdown systems are in pushdown control-flow
analysis~\cite{dvanhorn:Earl2010Pushdown}.
With respect to CFA2, the pushdown analysis presented here is
potentially polyvariant and in direct-style.

On the other hand, CFA2 distinguishes stack-allocated and
store-allocated variable bindings, whereas our formulation of pushdown
control-flow analysis does not and allocates all bindings in the
store.
If CFA2 determines a binding can be allocated on the stack, that
binding will enjoy added precision during the analysis and is not
subject to merging like store-allocated bindings.

Recently, Vardoulakis and Shivers have extended CFA2 to analyze
programs containing the control operator {\tt
  call-with-current-continuation}
\cite{dvanhorn:Vardoulakis2011Pushdown}.  
The operator, abbreviated {\tt call/cc}, works as follows: at the
point it is applied, it reifies the continuation as procedure; when
that procedure is applied it aborts the call's continuation and
installs the reified continuation.
CFA2 is able to analyze this powerful control operator, which is able
to encode all of the control operators considered here, but without
the same guarantees of precision as this work is able to provide for
the weaker notion of exceptions and breaks.  So while CFA2 can analyze
{\tt call/cc}, it does so with the potential for loss of precision
about the control stack; indeed, this appears to be inherently
necessary for any computable analysis of {\tt call/cc} as the operator
does not obey a stack discipline.  Vardoulakis has implemented CFA2
for JavaScript as the ``Doctor JS''
tool.\footnote{\tt{http://doctorjs.org/}}

The current work also draws on CFL- and pushdown-reachability
analysis~\cite{mattmight:Bouajjani:1997:PDA-Reachability,dvanhorn:Kodumal2004Set,mattmight:Reps:1998:CFL,mattmight:Reps:2005:Weighted-PDA}.
%
%
%
CFL-reachability techniques have also been used to compute classical
finite-state abstraction CFAs~\cite{mattmight:Melski:2000:CFL} and
type-based polymorphic control-flow
analysis~\cite{mattmight:Rehof:2001:TypeBased}.
\emph{These analyses should not be confused with pushdown control-flow
  analysis}: our results demonstrate how to compute a fundamentally
more precise kind of CFA, while the work on CFL-reachability has shown
how to cast classical analyses, such as 0CFA, as a reachability
problem for a context-free language.
%

\subsection{JavaScript analyses}

Thiemann~\cite{dvanhorn:Thiemann2005Towards} develops a type system for Core JavaScript,
a restricted subset of JavaScript.
The type system rules out the application of non-functions, applying primitive
operations to values of incorrect base type, dereferencing fields of the {\tt
undefined} or {\tt null} value, and referring to unbound variables.
Jensen, M\o ller, and Thiemann, \cite{dvanhorn:Jensen2009Type} develop
an abstract interpretation computing type inference.
It builds on the type system of
Thiemann~\cite{dvanhorn:Thiemann2005Towards}, using it as inspiration for their
abstract domains.  
%
%
Richards \etal's landmark empirical survey of JavaScript
code~\cite{dvanhorn:Richards2010Analysis} made it clear that for
JavaScript analyses to work in the wild, it is not sufficient to
handle only a well-behaved core of the language.
Capturing ill-behaved parts of JavaScript soundly and precisely was a
major motivation for our research.

Subsequently, Heidegger and Thiemann have extended the type system
with a notion of \emph{recency} to improve
precision~\cite{dvanhorn:Heidegger2010Recency} and Jensen~\etal, have
developed a technique of \emph{lazy propagation} to increase the
feasibility of the analysis~\cite{dvanhorn:Jensen2011Interprocedural}.
Balakrishnan and Reps pioneered the idea of recency~\cite{mattmight:Balakrishnan:2006:Recency}:
objects are considered recent when created and given a singleton type
and treated flow-sensitively until ``demoted'' to a summary type that
is treated flow-insensitively.
Recency
enables strong update in analyses~\cite{dvanhorn:jagannathan-etal-popl98}, which is important for reasoning
precisely about initialization patterns in JavaScript programs.
Recency and lazy propagation are orthogonal to our analytic framework:
in our recent work, we show how to incorporate a generalization
of recency into a machine-based static analysis~\cite{mattmight:Might:2010:Shape}
through the concept of anodization.

Guha, Krishnamurthi and Jim~\cite{dvanhorn:Guha2009Using} developed an analysis
for intrusion-detection of JavaScript, driven in part by an adaptation of
  $k$-CFA to a large subset of JavaScript.
Our work differs from their work in that we are formally guaranteeing
soundness with respect to a concrete semantics, we provide
fine-grained control over precision for our finite-state analysis and
we also provide a pushdown analysis for handling the complex non-local
control features which pervade JavaScript code.
(Guha~\etal\ make a best-effort attempt at soundness and dynamically
detect violations of soundness in empirical trials, violations which they use
to refine their analysis.)

Chugh~\etal~\cite{dvanhorn:Chugh2009Staged} present a staged information-flow
analysis of JavaScript.
In effect, their algorithm partially evaluates the analysis with respect to the
available JavaScript to produce a residual analysis.
When more code becomes available, the residual analysis resumes.
Our own framework is directly amenable to such partial evaluation for handling
constructs like {\tt eval}: explore the state-space aggressively, but do not
explore past {\tt eval} states.
The resulting partial abstract transition graph is sound until the program
encounters {\tt eval}.
At this point, the analysis may be resumed with the code supplied to {\tt
eval}.


\section{Conclusions and perspective}
\label{sec:conclusion}

We present a principled systematic derivation of machine-based analysis for
JavaScript.
By starting with an established formal semantics and transforming it into an
abstract machine, we soundly capture JavaScript in full, quirks and all.
The abstraction of this machine yields a robust finite-state framework for the
static analysis of JavaScript, capable of instantiating the equivalent of
traditional techniques such as $k$-CFA and CPA.
Finding the traditional finite-state approach wanting in precision for
JavaScript's extensive use of non-local control, we extend the theory of
systematic abstraction of abstract machines from finite-state to pushdown.
These decidable pushdown machines precisely model the structure of the program
stack, and so do not lose precision in the presence of control constructs that
depend on it, such as  recursion or complex exceptional control-flow.

\begin{center}
{\tt https://github.com/dvanhorn/jam/}
\end{center}

\paragraph{Acknowledgments} We thank Arjun Guha for answering questions about $\lambda_\JS$.
Sam Tobin-Hochstadt and Mitchell Wand provided fruitful discussions
and feedback on early versions of this work.  We thank Olivier Danvy
and Jan Midtgaard for their hospitality and the rich intellectual
environment they provided during the authors' visit to Aarhus
University.



\bibliographystyle{plain}
\bibliography{bibliography}

\appendix

\end{document}